\documentclass{llncs}
\usepackage{hyperref}
\usepackage{textcomp}
\usepackage{graphicx}
\usepackage[T1]{fontenc}
\usepackage{amsmath}
\usepackage{amsfonts}
\usepackage{amssymb}

\newcommand{\comment}[1]{}

\def\dqt{\textquotedbl}

\newcommand{\evalexpr}[1]{\left[\!\left[ #1 \right]\!\right] }
\newcommand{\transexpr}[1]{\left[ #1 \right]_{\cal R} }
\newcommand{\iri}[1]{<\!{\tt{#1}}\!>}

\def\utup{\mbox{$U$-{\tt Tup}}}

\begin{document}

\long\def\comment#1{}

\title{Provenance for SPARQL queries}

\author{C. V. Dam\'asio\inst{1} and A. Analyti\inst{2} and G. Antoniou\inst{3}}

\institute{CENTRIA, Departamento de Inform\'atica
Faculdade de Ci\^encias e Tecnologia
Universidade Nova de Lisboa,
2829-516 Caparica, Portugal.\\
\email{cd@fct.unl.pt}
\and Institute of Computer Science, FORTH-ICS,
Crete, Greece \\
\email{analyti@ics.forth.gr}
\and 
Institute of Computer Science, FORTH-ICS, and \\
Department of Computer Science, University of Crete,
Crete, Greece \\
\email{antoniou@ics.forth.gr}
}

\date{March 2013}

\footnotetext[0]{Extended version of the ISWC 2012 paper including proofs, published in Lecture Notes on Computer Science 7649, pp. 625-640, ISBN 978-3-642-35175-4, Springer Berlin Heidelberg, 2012. The original publication is available at \href{http://dx.doi.org/10.1007/978-3-642-35176-1_39}{www.springerlink.com}.}

\maketitle

\begin{abstract}
Determining trust of data available in the Semantic Web is fundamental for applications and users, in particular for linked open data obtained from SPARQL endpoints. There exist several proposals in the literature to annotate SPARQL query results with values from abstract models, adapting the seminal works on provenance for annotated relational databases. We provide an approach capable of providing provenance information for a large and significant fragment of SPARQL 1.1, including for the first time the major non-monotonic constructs under multiset semantics. The approach is based on the translation of 
SPARQL into relational queries over annotated relations with values of the most general m-semiring, and in this way also refuting a claim in the literature that the {\tt OPTIONAL} construct of SPARQL cannot be captured appropriately with the known abstract models. 
\end{abstract}

\begin{keywords}
How-provenance, SPARQL queries, m-semirings, difference
\end{keywords}

\section{Introduction}

A general data model for annotated relations has been introduced in~\cite{GKT07:pods}, for positive relational algebra (i.e. excluding the difference operator).
These annotations can be used to check derivability of a tuple, lineage, and provenance, perform query evaluation of incomplete database, etc. 
The main concept is the notion of ${\cal K}$-relations where tuples are annotated with values (tags) of a commutative semiring $\cal K$\@, while positive relational algebra operators semantics
are extended and captured by corresponding compositional operations over $\cal K$\@. 
The obtained algebra on $\cal K$-relations is expressive enough to capture different kinds of annotations with set or bag semantics, 
and the authors show that the semiring of polynomials with integer coefficients is the most general semiring. 
This means that to evaluate queries for any positive algebra query on an arbitrary semiring, one can evaluate the query in the semiring of polynomials (factorization property of~\cite{GKT07:pods}).
This work has been extended to the case of full relational algebra in~\cite{GeertsP10} by considering the notion of semirings with a monus operation ($m$-semirings~\cite{amer84}) and constant annotations, and the factorization property is  proved for the special $m$-semiring that we denote by ${\cal K}_{dprovd}$\@. 

The use of these abstract models based on $\cal K$-relations to express provenance in the Semantic Web has been advocated in~\cite{Theoharis:2011}.
However, the authors claim that the existing $m$-semirings are not capable to capture the appropriate provenance information for SPARQL queries.
This claim is supported by the authors using a simple example, which we have adapted to motivate our work:

\begin{example}\label{ex:motopt} Consider the following RDF graph expressing information about users' accounts and homepages, resorting to the FOAF vocabulary:
{\small
\begin{verbatim}
   @prefix people: <http://people/> .
   @prefix foaf: <http://xmlns.com/foaf/0.1/> .
   people:david foaf:account <http://bank> .
   people:felix foaf:account <http://games> .
   <http://bank> foaf:accountServiceHomepage <http://bank/yourmoney>.
\end{verbatim}}
\noindent The SPARQL query
{\small
\begin{verbatim}
   PREFIX foaf <http://xmlns.com/foaf/0.1/>
   SELECT *
   WHERE { ?who foaf:account ?acc .
           OPTIONAL { ?acc foaf:accountServiceHomepage ?home }
   }  
\end{verbatim}}
\noindent returns the solutions (mappings of variables):

{\tt\small
\begin{tabular}{|l|l|lc|}\hline
?who & ?acc & ?home & \\\hline\hline
<http://people/david> & <http://bank> & <http://bank/yourmoney> &\\\hline
<http://people/felix> & <http://games> & & \\\hline
\end{tabular}\\
}\\
However, if the last triple is absent from the graph then the solutions are instead:\\

{\tt\small
\begin{tabular}{|l|l|lc|}\hline
?who & ?acc & ?home & \\\hline\hline
<http://people/david> & <http://bank> &  &\\\hline
<http://people/felix> & <http://games> & & \\\hline
\end{tabular}\\
}
\end{example}

In order to track provenance of data, each tuple of data can be tagged with an annotation of a semiring. This annotation can be a boolean, e.g. to annotate that the tuple is trusted or not, a set of identifiers of tuples returning lineage of the tuple, or more complex annotations like the polynomials semiring to track full how-provenance~\cite{GKT07:pods,GeertsP10}\@, i.e. how a tuple is generated in the result under bag semantics. \comment{The semantics of standard positive algebra operators (selection $\sigma$, projection $\pi$, union $\cup$, natural join $\bowtie$, and renaming $\rho$) can be naturally extended~\cite{GKT07:pods} to $\cal K$-relations where, briefly, unions are encoded as sums ($\oplus)$ in $\cal K$\@, and joins as products ($\otimes$).}

Returning to the introductory example, assume that we represent the 3 triples in the input RDF graph as the ternary $K_{dprovd}$-relation (with the obvious abbreviations), where the last column contains the triple identifier (annotation):
\[\small
\begin{array}{l@{\quad}ll||l}
\multicolumn{4}{l}{\tt Triples}\\
 {\tt sub} & {\tt pred} & {\tt obj} & \\\hline
\iri{david} & \iri{account} & \iri{bank} & t_1\\
\iri{felix} & \iri{account} & \iri{games} & t_2\\
\iri{bank} & \iri{accountServiceHomepage} & \iri{bank/yourmoney} & t_3\\
\end{array}\\
\]
The expected annotation of the first solution of the SPARQL query is $t_1 \times t_3$\@, meaning that the solution was obtained by joining triples identified by $t_1$ and $t_3$\@, while for the second solution the corresponding annotation is simply $t_2$\@. However, if we remove the last tuple we obtain a different solution for {\tt david} with annotation just $t_1$\@. The authors in~\cite{Theoharis:2011} explain why the existing approaches to provenance for the Semantic Web cannot handle the situation of Example~\ref{ex:motopt}, basically because there are different bindings of variables depending on the absence/presence of triples, and it is claimed that the $m$-semiring $K_{dprovd}$ also cannot handle it. The rest of our paper shows that this last claim is wrong, but that requires some hard work and long definitions since the method proposed relies on the translation of SPARQL queries into relational algebra. The result is the first approach that provides adequate provenance information for {\tt OPTIONAL}, {\tt MINUS} and {\tt NOT EXISTS} constructs under the multiset (bag) semantics of SPARQL.

The organization of the paper is the following. We review briefly in the next section the basics of $K$-relations\@. The SPARQL semantics is introduced in Section~\ref{sec:sparql}, and its translation into relational algebra is the core of the paper and can be found in Section~\ref{sec:sparql2ra}. Using the relational algebra translation of SPARQL, we use $K_{dprovd}$ to annotate SPARQL queries and show in Section~\ref{sec:prov} that Example~\ref{ex:motopt} is properly handled. We finish with some comparisons and conclusions.

\section{Provenance for $K$-relations}

A commutative semiring is an algebraic structure ${\cal K} = (\mathbb{K},\oplus,\otimes,0,1)$ where $(\mathbb{K},\oplus,0)$ is a commutative monoid ($\oplus$ is associative and commutative) with identity element $0$, $(\mathbb{K},\otimes,1)$ is a commutative monoid with identity element $1$, the operation $\otimes$ distributes over $\oplus$, and $0$ is the annihilating element of $\otimes$\@. In general, a tuple is a function $t : U \rightarrow \mathbb{D}$ where $U$ is a finite set of attributes and $\mathbb{D}$ is the domain of values, which is assumed to be fixed. The set of all such tuples is \utup\ and usual relations are subsets of \utup\@.  A $\cal K$-relation over $U$ is a function $R:\utup \rightarrow \mathbb{K}$, and its support is $supp(R) = \{t \mid R(t) \not= 0 \}$\@.  

In order to cover the full relational operators, the authors in~\cite{GeertsP10} assume that the $\cal K$ semiring is naturally ordered (i.e. binary relation $x \preceq y$ is a partial order, where $x \preceq y$ iff there exists $z \in \mathbb{K}$ such that $x \oplus z = y$ ), and require additionally that for every pair $x$ and $y$ there is a least $z$ such that $x \preceq y \oplus z$, defining in this way $x \ominus y$ to be such smallest $z$\@. A $\cal K$ semiring with such a monus operator is designated by $m$-semiring. Moreover, in order to capture duplicate elimination, the authors assume that the $m$-semiring is finitely generated. The query language\footnote{The authors use instead the notation  $\cal{RA}^+_{\cal K}(\setminus,\delta)$\@.} $\cal{RA}^+_{\cal K}(-,\delta)$ has the following operators~\cite{GeertsP10}:
\begin{description}
\item[empty relation:] For any set of attributes $U$, we have $\emptyset:\utup \rightarrow {\mathbb K}$ such that $\emptyset(t)=0$ for any $t$\@.
\item[union] If $R_1,R_2:\utup\rightarrow{\mathbb K}$ then $R_1\cup R_2: \utup \rightarrow {\mathbb K}$ is defined by:\\ $(R_1\cup R_2)(t)=R_1(t)\oplus R_2(t)$\@.
\item[projection] If $R:\utup\rightarrow{\mathbb K}$ and $V \subseteq U$ then $\Pi_V(R):V$-${\tt Tup} \rightarrow{\mathbb K}$ is defined by
$\left(\Pi_V(R)\right)(t) = \bigoplus_{t=t' \text{ on } V \text{ and } R(t') \not= 0} R(t')$\@.
\item[selection:] If $R:\utup\rightarrow{\mathbb K}$  and the selection predicate $P$ maps each $U$-tuple to either $0$ or $1$ depending on
the (in-)equality of attribute values, then $\sigma_P(R):\utup\rightarrow{\mathbb K}$ is defined by $\left(\sigma_P(R)\right)(t)=R(t) \otimes P(t)$\@.
\item[natural join] If $R_i:\mbox{$U_i$-{\tt Tup}}\rightarrow{\mathbb K}$, for $i=1,2$\@,  then $R_1\Join ‰R_2$ is the $\cal K$-relation over $U_1 \cup U_2$ defined by $(R_1 \Join R_2)(t) = R_1(t) \otimes R_2(t)$\@.
\item[renaming] If $R:\utup\rightarrow{\mathbb K}$ and $\beta:U \rightarrow U'$ is a bijection then $\rho_\beta(R)$ is the $\cal K$-relation over $U'$ defined by $(\rho_\beta(R))(t) = R(t \circ \beta^{-1})$\@.
\item[difference:] If $R_1,R_2:\utup\rightarrow{\mathbb K}$ then $R_1- R_2: \utup \rightarrow {\mathbb K}$ is defined by:\\ $(R_1- R_2)(t)=R_1(t)\ominus R_2(t)$\@.
\item[constant annotation:] If $R:\utup\rightarrow{\mathbb K}$  and $k_i$ is a generator of $\mathbb K$ then $\delta_{k_i}:\utup\rightarrow{\mathbb K}$ is defined by $(\delta_{k_i}(R))(t)=k_i$ for each $t \in supp(R)$ and  $(\delta_{k_i}(R))(t)=0$ otherwise.
\end{description}

One major result of~\cite{GeertsP10} is that the factorization property can be obtained for $\cal{RA}^+_{\cal K}(-,\delta)$ by using a special $m$-semiring with constant annotations that we designate by ${\cal K}_{dprovd}$\@.  ${\cal K}_{dprovd}$ is the free $m$-semiring over the set of source tuple ids $X$\@, which is a free algebra generated by the set of (tuple) identifiers in the equational variety of $m$-semirings. Elements of ${\cal K}_{dprovd}$ are therefore terms defined inductively as: identifiers in $X$, $0$, and $1$ are terms;  if $s$ and $t$ are terms then $(s+t)$, $(s \times t)$\@, $(s-t)$\@, and $\delta_{k_i}(t)$ are terms, and nothing else is a term. In fact annotations of  ${\cal K}_{dprovd}$ are elements of the quotient structure of the free terms with respect to the congruence relation induced by the axiomatization of the $m$-semirings, in order to guarantee the factorization property  (see~\cite{GeertsP10} for more details). In our approach, $X$ will be the set of graph and tuple identifiers.

We slightly extend the projection operator, by introducing new attributes whose value can be computed from the other attributes. In our approach, this is simply syntactic sugar since the functions we use are either constants or return one of the values in the arguments.

\section{SPARQL semantics}\label{sec:sparql}

The current draft of SPARQL 1.1~\cite{sparql11:rec} defines the semantics of SPARQL queries via a translation into SPARQL algebra operators, which are then evaluated with respect to a given RDF dataset. In this section, we overview an important fragment corresponding to an extension of the work in~\cite{Perez:2009} that presents the formal semantics of the first version of SPARQL. The aim of our paper is on the treatment of non-monotonic constructs\footnote{OPTIONAL is non-monotonic when combined with some FILTER expressions,  or if we understand unbound variables incomparable to bounded ones like we assume in the translation.} of SPARQL, namely OPTIONAL, MINUS and NOT EXISTS, and thus we focus in the SELECT query form, ignoring property paths, GROUP graph patterns and aggregations, as well as solution modifiers. The extension of our work to consider all the graph patterns is direct from the results presented. Regarding FILTER expressions, we analyse with detail the EXISTS and NOT EXISTS constructs, requiring special treatment.
We assume the reader has basic knowledge of RDF and we follow closely the presentation of~\cite{sparql11:rec}. For more details the reader is referred to sections 17 and 18 of the current SPARQL 1.1 W3C recommendation. We also make some simplifying assumptions that do not affect the results of our paper. 

\subsection{Basics}

Consider disjoint sets of IRI (absolute) references {\bf I}, blank nodes {\bf B}, and literals {\bf L} including plain literals and typed literals, and an infinite set of variables {\bf V}. The set of RDF terms is $\mathbf{T} = \mathbf{IBL} = {\mathbf I} \cup {\mathbf B} \cup {\mathbf L}$\@. A triple\footnote{Literals in the subject of triples are allowed, since this generalization is expected to be adopted in the near future.} $\tau = (s,p,o)$ is an element of $\mathbf{IBL} \times {\mathbf I} \times \mathbf{IBL}$ and a graph is a set of triples. Queries are evaluated with respect to a given RDF Dataset $D=\{ G, (\iri{u_1}, G_1), (\iri{u_2}, G_2), \ldots, (\iri{u_n}, G_n) \}$\@, where $G$ is the default graph, and each pair $(\iri{u_i}, G_i)$ is called a named graph, with each IRI $\tt u_i$  distinct in the RDF dataset, and $G_i$ being a graph.

\subsection{Graph patterns}

SPARQL queries are defined by graph patterns, which are obtained by combining triple patterns with operators. SPARQL graph patterns are defined recursively as follows:
\begin{itemize}
\item The empty graph pattern $()$\@.
\item A tuple $(\mathbf{IL} \cup \mathbf{V} ) \times (\mathbf{I} \cup \mathbf{V} ) \times (\mathbf{IL} \cup \mathbf{V} )$ is a graph pattern called triple pattern\footnote{For simplicity, we do not allow blank nodes in triple patterns.};
\item If $P_1$ and $P_2$ are graph patterns then $(P_1 {\tt\ AND\ } P_2)$\@, $(P_1 {\tt\ UNION\ } P_2)$\@,  as well as $(P_1 {\tt\ MINUS\ } P_2)$\@, and $(P_1 {\tt\ OPTIONAL\ } P_2)$\@ are graph patterns;
\item If $P_1$ is a graph pattern and $R$ is a filter SPARQL expression\footnote{For the full syntax of filter expressions, see the W3C recommendation~\cite{sparql11:rec}.} then  the construction $(P_1 {\tt\ FILTER\ } R)$ is a graph pattern;
\item If $P_1$ is a graph pattern and $term$ is a variable or an IRI then $({\tt GRAPH\ } term\ P_1)$ is a graph pattern.
\end{itemize}

The SPARQL 1.1 Working Draft also defines Basic Graph Patterns (BGPs), which correspond to sets of triple patterns. A Basic Graph Pattern ${P_1, \ldots, P_n}$ is encoded as the graph pattern $(() {\tt\ AND\ } (P_1  {\tt\ AND\ } (P_2 \ldots {\tt\ AND\ } P_n))\ldots)$\@. We ignore in this presentation the semantics of {\tt FILTER} expressions, whose syntax is rather complex. For the purposes of this paper it is enough to consider that these expressions after evaluation return a boolean value, and therefore we also ignore errors. However, we show how to treat the {\tt EXISTS} and {\tt NOT EXISTS} patterns in {\tt FILTER} expressions since these require querying graph data, and therefore provenance information should be associated to these patterns.

\subsection{SPARQL algebra}

Evaluation of SPARQL patterns return multisets (bags) of solution mappings. A solution mapping, abbreviated solution, is a partial function $\mu: \mathbf{V} \rightarrow \mathbf{T}$\@. The domain of $\mu$ is the subset of variables of $\mathbf{V}$ where $\mu$ is defined. Two mappings $\mu_1$ and $\mu_2$ are compatible if for every variable $v$ in $dom(\mu_1) \cap dom(\mu_2)$ it is the case that $\mu_1(v) = \mu_2(v)$\@. It is important to understand that any mappings with disjoint domain are compatible, and in particular the  solution mapping $\mu_0$ with empty domain is compatible with every solution. If two solutions $\mu_1$ and $\mu_2$ are compatible then their union $\mu_1 \cup \mu_2$ is also a solution mapping. We represent extensionally a solution mapping as a set of pairs of the form $(v,t)$\@; in the case of a solution mapping with a singleton domain we use the abbreviation $v \rightarrow t$\@.
Additionally, if $P$ is an arbitrary pattern we denote by $\mu(P)$ the result of substituting the variables in $P$ defined in $\mu$ by their assigned values.

 We denote that solution mapping $\mu$ satisfies the filter expression $R$ with respect to the active graph $G$ of dataset $D$ by $\mu \models_{D(G)} R$\@. Including the parameter $D(G)$ in the evaluation of filter expressions is necessary in order to evaluate ${\tt EXISTS}(P)$ and ${\tt NOT\ EXISTS}(P)$ filter expressions, where $P$ is an arbitrary graph pattern. If these constructs are removed from the language, then one only needs to consider the current solution mapping to evaluate expressions (as done in~\cite{Perez:2009}).

\begin{definition}[SPARQL algebra operators~\cite{sparql11:rec}]
Let $\Omega_1$ and $\Omega_2$ be multisets of solution mappings, and $R$ a filter expression. Define:
\begin{description}
\item[Join:] $\Omega_1 \bowtie \Omega_2 = \{\!|\mu_1 \cup \mu_2 \mid \mu_1 \in \Omega_1 \text{ and }  \mu_2 \in \Omega_2 \text{ such that } \mu_1 \text{ and } \mu_2$ $\text{ are compatible }|\!\}$
\item[Union:] $\Omega_1 \cup \Omega_2 = \{\!|\mu \mid \mu \in \Omega_1 \text{  or }  \mu \in \Omega_2|\!\}$
\item[Minus:] $\Omega_1 - \Omega_2 = \{\!|\mu_1  \mid \mu_1 \in \Omega_1 \text{ such that }  \forall_{\mu_2 \in \Omega_2}  \text{ either } \mu_1 \text{ and } \mu_2$ are not compatible $\text{or } dom(\mu_1) \cap dom(\mu_2) = \emptyset |\!\}$
\item[Diff:] $\Omega_1 \setminus^{D(G)}_R\ \Omega_2 = \{\!|\mu_1  \mid \mu_1 \in \Omega_1 \text{ such that }  \forall_{\mu_2 \in \Omega_2}  \text{ either } \mu_1 \text{ and } \mu_2$ are not compatible, $\text{or } \mu_1 \text{ and } \mu_2 \text{ are compatible and  } \mu_1 \cup \mu_2 \not\models_{D(G)} R  |\!\}$
\item[LeftJoin:] $\Omega_1 \sqsupset\!\bowtie^{D(G)}_R \Omega_2 = (\Omega_1 \bowtie \Omega_2) \cup (\Omega_1 \setminus^{D(G)}_R \Omega_2)$
\end{description}
\end{definition}

 The {\bf Diff} operator is auxiliary to the definition of {\bf LeftJoin}. The SPARQL 1.1 Working Draft also introduces the notion of sequence to provide semantics to modifiers like {\tt ORDER BY}. The semantics of the extra syntax is formalized by several more operators, namely aggregates and sequence modifiers (e.g. ordering), as well as property path expressions; we briefly discuss their treatment later on. Since lists can be seen as multisets with order and, without loss of generality regarding provenance information, we just consider multisets. 

\begin{definition}[SPARQL graph pattern evaluation] Let $D(G)$ be an RDF dataset with active graph $G$\@, initially the default graph in $D(G)$\@. Let $P$, $P_1$ and $P_2$ be arbitrary graph patterns, and $t$ a triple pattern. The evaluation of a graph pattern over $D(G)$\@, denoted by $\evalexpr{.}_{D(G)}$ is defined recursively as follows:
\begin{enumerate}
\item $\evalexpr{()}_{D(G)} = \{\!| \mu_0 |\!\}$\@;
\item $\evalexpr{t}_{D(G)} = \{\!|\ \mu \mid dom(\mu) = var(t) \text{ and } \mu(t) \in G |\!\}$\@, where $var(t)$ is the set of variables occurring in the triple pattern $t$\@;
\item $\evalexpr{(P_1  {\tt\ AND\ } P_2)}_{D(G)} = \evalexpr{P_1}_{D(G)} \bowtie  \evalexpr{P_2}_{D(G)}$\@;
\item $\evalexpr{(P_1  {\tt\ UNION\ } P_2)}_{D(G)} = \evalexpr{P_1}_{D(G)} \cup  \evalexpr{P_2}_{D(G)}$\@;
\item $\evalexpr{(P_1  {\tt\ MINUS\ } P_2)}_{D(G)} = \evalexpr{P_1}_{D(G)} -  \evalexpr{P_2}_{D(G)}$\@;
\item $\evalexpr{(P_1  {\tt\ OPTIONAL\ } P_2)}_{D(G)} = \evalexpr{P_1}_{D(G)} \sqsupset\!\bowtie^{D(G)}_{true}  \evalexpr{P_2}_{D(G)}$\@, where $P_2$ is not a {\tt FILTER} pattern;
\item $\evalexpr{(P_1  {\tt\ OPTIONAL\ } (P_2 {\tt\ FILTER\ } R))}_{D(G)} = \evalexpr{P_1}_{D(G)} \sqsupset\!\bowtie^{D(G)}_R  \evalexpr{P_2}_{D(G)}$\@;
\item $\evalexpr{(P_1  {\tt\ FILTER\ } R)}_{D(G)} = \{\!|\ \mu \in \evalexpr{P_1}_{D(G)} \mid \mu \models_{D(G)} R\ |\!\}$\@;
\item Evaluation of $\evalexpr{({\tt GRAPH\ } term\ P_1)}_{D(G)}$  depends on the form of $term$\@:
\begin{itemize}
\item If $term$ is an IRI corresponding to a graph name ${\tt u_i}$ in $D(G)$ then $\evalexpr{({\tt GRAPH\ } term\ P_1)}_{D(G)} = \evalexpr{P_1}_{D(G_i)}$\@;
\item If $term$ is an IRI that does not correspond to any graph in $D(G)$ then $\evalexpr{({\tt GRAPH\ } term\ P_1)}_{D(G)} = \{\!| |\!\}$\@;
\item If $term$ is a variable $v$ then $\evalexpr{({\tt GRAPH\ } term\ P_1)}_{D(G)} =$
\[ =  ( \evalexpr{P_1}_{D(G_1)} \bowtie \{\!| v \rightarrow \iri{u_1} |\!\}) \cup \ldots \cup ( \evalexpr{P_1}_{D(G_n)} \bowtie \{\!| v \rightarrow \iri{u_n} |\!\} ) \]
\end{itemize}
\end{enumerate}
\end{definition}

The evaluation of {\tt EXISTS} and {\tt NOT EXISTS} is performed in the satisfies relation of filter expressions. 

\begin{definition} Given a solution mapping $\mu$ and a graph pattern $P$  over an RDF dataset $D(G)$ then $\mu \models_{D(G)} {\tt\ EXISTS} (P)$ (resp. $\mu \models_{D(G)} {\tt\ NOT\ EXISTS}(P)$) iff $\evalexpr{\mu(P)}_{D(G)}$ is a non-empty (resp. empty) multiset.
\end{definition}
           
\begin{example}\label{ex:query} The SPARQL query of Example~\ref{ex:motopt} corresponds to the following graph pattern :
\[
\begin{array}{rl}
Q = ( & (?who, \iri{foaf:account}, ?acc ) {\tt\ OPTIONAL\ }\\
  & (?acc, \iri{foaf:accountServiceHomepage}, ?home )\\
) &\\
\end{array}
\]
The evaluation of the query result with respect to the RDF dataset $D=\{G\}$, just containing the default graph $G$, specified in the example is:
\[\small
\begin{array}{l}
\evalexpr{Q}_{D(G)} = \\
\ \evalexpr{(?who, \iri{foaf:account}, ?acc )}_{D(G)} \sqsupset\!\bowtie^{D(G)}_{true}\\
\ \evalexpr{(?acc, \iri{foaf:accountServiceHomepage}, ?home )}_{D(G)}\\
= \{\!| \{(?who,\iri{http://people/david}),(?acc,\iri{http://bank}) \}, \\
\ \ \ \ \ \{(?who,\iri{http://people/felix}),(?acc,\iri{http://games}) \} |\!\}  \sqsupset\!\bowtie^{D(G)}_{true}\\
\ \ \ \{\!| \{(?acc,\iri{http://bank}),(?home,\iri{http://bank/yourmoney}) \}\ |\!\}\\
= \{\!| \{(?who,\iri{http://people/david}),(?acc,\iri{http://bank}),\\
\ \ \ \ \ \ \ (?home,\iri{http://bank/yourmoney}) \}, \\
\ \ \ \ \ \{(?who,\iri{http://people/felix}),(?acc,\iri{http://games}) \}\\
\ \ \ |\!\}\\
\end{array}
\]
The evaluation of query $Q$ returns, as expected, two solution mappings.
\end{example}

\section{Translating SPARQL algebra into relational algebra}\label{sec:sparql2ra}

The rationale for obtaining how-provenance for SPARQL is to represent each solution mapping as a tuple of a relational algebra query  constructed from the original SPARQL graph pattern. The construction is intricate and fully specified, and is inspired from the translation of full SPARQL 1.0 queries into SQL, as detailed in~\cite{Elliott:2009}, and into Datalog in~\cite{Polleres07}. Here, we follow a similar strategy but for simplicity of presentation we assume that a given RDF dataset $D=\{ G_0, (\iri{u_1}, G_1), (\iri{u_2}, G2), \ldots, (\iri{u_n}, G_n) \}$  is represented by the two relations: {\tt Graphs(gid,IRI)}  and  {\tt Quads(gid,sub,pred,obj)}\@. The former stores information about the graphs in the dataset $D$ where {\tt gid} is a numeric graph identifier, and {\tt IRI} an IRI reference. The relation {\tt Quads} stores the triples of every graph in the RDF dataset. Different implementations may  immediately adapt the translation provided here in this section to their own schema.

Relation {\tt Graphs(gid,IRI)}  contains a tuple $(i,\iri{u_i})$ for each named graph $(\iri{u_i}, G_i)$, and the tuple $(0,\iri{})$ for the default graph, while relation {\tt Quads(gid,sub,pred,obj)} stores a tuple of the form $(i,s,p,o)$ for each triple $(s,p,o) \in G_i$\footnote{For simplicity {\tt sub}, {\tt pred}, and {\tt obj} are text attributes storing lexical forms of the triples' components. We assume that datatype literals have been normalized, and blank nodes are distinct in each graph. The only constraint is that different RDF terms must be represented by different strings; this can be easily guaranteed.}. With this encoding, the default graph  always has identifier $0$\@, and all the graph identifiers are consecutive integers.

It is also assumed the existence of a special value {\tt unb}, distinct from the encoding of any RDF term, to represent that a particular variable is unbound in the solution mapping. This is required in order to be able to represent solution mappings as tuples with fixed and known arity. Moreover, we assume that the variables are totally ordered (e.g. lexicographically). The translation requires the full power of relational algebra, and notice that bag semantics is assumed (duplicates are allowed) in order to obey to the cardinality restrictions of SPARQL algebra operators~\cite{sparql11:rec}.

\begin{definition}[Translation of triple patterns] Let $t=(s,p,o)$ be a triple pattern and $G$ an attribute. Its translation $\transexpr{(s,p,o)}^G$ into relational algebra is constructed from relation {\tt Quads} as follows:
\begin{enumerate}
\item Select the tuples with the conjunction obtained from the triple pattern by letting ${\tt Quads.sub} = s$ (resp. ${\tt Quads.pred} = p$, ${\tt Quads.obj} = o$) if $s$ (resp. $p$, $o$) are RDF terms; if a variable occurs more than once in $t$\@, then add an equality condition among the corresponding columns of {\tt Quads};
\item Rename {\tt Quads.gid} as $G$; rename as many as {\tt Quads} columns as distinct variables that exist in $t$\@, such that there is exactly one renamed column per variable;
\item Project in $G$ and variables occurring in $t$\@;
\end{enumerate}
The empty graph pattern is translated as $\transexpr{()}^G = \Pi_G\left[  \rho_{G \leftarrow {\tt gid}} ({\tt Graphs}) \right]$\@.
\end{definition}

\begin{example} Consider the following triple patterns:
\[
\begin{array}{l}
t_1 = (?who, \iri{http://xmlns.com/foaf/0.1/account}, ?acc )\\
t_2 = (?who, \iri{http://xmlns.com/foaf/0.1/knows}, ?who )\\
t_3 = (\iri{http://cd}, \iri{http://xmlns.com/foaf/0.1/name}, \text{\dqt Carlos\dqt@pt} )\\
\end{array}
\]
The corresponding translations into relational algebra are:
\[
\begin{array}{lll}
\transexpr{t_1}^G & = & \Pi_{G,acc,who}\left[ \rho_{\scriptsize\begin{array}{l}G \leftarrow {\tt gid}\\ acc \leftarrow {\tt obj}\\  who \leftarrow {\tt sub}\end{array}} \left({\scriptsize\sigma_{{\tt pred} = \iri{http://xmlns.com/foaf/0.1/account} }}({\tt Quads}) \right) \right]\\
\transexpr{t_2}^G & = & \Pi_{G,who}\left[ \rho_{\scriptsize\begin{array}{l}G \leftarrow {\tt gid}\\ who \leftarrow {\tt sub}\end{array}} \left(\sigma_{\scriptsize\begin{array}{c}{\tt pred} = \iri{http://xmlns.com/foaf/0.1/knows}\\ \wedge\\ {\tt sub } = {\tt obj} \end{array} }({\tt Quads}) \right) \right]\\
\transexpr{t_3}^G & = & \Pi_{G}\left[ \rho_{\scriptsize\begin{array}{l}G \leftarrow {\tt gid}\end{array}} \left(\sigma_{\scriptsize\begin{array}{c}{\tt sub} = \iri{http://cd} \wedge\\ {\tt pred} = \iri{http://xmlns.com/foaf/0.1/name} \wedge\\ {\tt obj } = \text{\dqt Carlos\dqt@pt}\end{array} }({\tt Quads}) \right) \right]\\
\end{array}
\]
\end{example}

The remaining pattern that requires querying base relations is {\tt GRAPH}:
\begin{definition}[Translation of {\tt GRAPH} pattern] Consider the graph pattern $({\tt GRAPH\ } term\ P_1)$ and let $G'$ be a new attribute name.
\begin{itemize}
\item If $term$ is an IRI then $\transexpr{({\tt GRAPH\ } term\ P_1)}^G$ is
\[
\transexpr{()}^G \Join \Pi_{var(P_1)} \left[ \Pi_{G'} \left( \rho_{\scriptsize G' \leftarrow {\tt gid}} \left( \sigma_{term = {\tt IRI}}( {\tt Graphs}) \right) \right)\Join \transexpr{P_1}^{G'} \right]
\]
\item If $term$ is a variable $v$ then $\transexpr{({\tt GRAPH\ } term\ P_1)}^G$ is
\[
\transexpr{()}^G \Join \Pi_{\{v\} \cup var(P_1)} \left[ \rho_{\scriptsize G' \leftarrow {\tt gid}, v \leftarrow {\tt IRI}} \left( \sigma_{{\tt gid} > 0}( {\tt Graphs}) \right)\Join \transexpr{P_1}^{G'} \right]
\]
\end{itemize}
\end{definition}

Notice that the relational algebra query resulting from the translation of the pattern graph $P_1$ renames and hides the graph attribute. The join of the empty pattern is included in order to guarantee that each  query returns  the graph identifier in the first ``column''.

\begin{definition}[Translation of the {\tt UNION} pattern] Consider the graph pattern $(P_1  {\tt\ UNION\ } P_2)$\@. The relation algebra expression $\transexpr{(P_1  {\tt\ UNION\ } P_2)}^G$  is:
\[
\begin{array}{c}
\Pi_{G,var(P_1) \cup \{ v \leftarrow {\tt unb} \mid v \in var(P_2) \setminus var(P_1)\}} \left( \transexpr{P_1}^G \right)\\
\bigcup\\
\Pi_{G,var(P_2) \cup \{ v \leftarrow {\tt unb} \mid v \in var(P_1) \setminus var(P_2)\}} \left( \transexpr{P_2}^G \right)\\
\end{array}
\]  
\end{definition}

The union operator requires the use of an extended projection in order to make unbound variables which are present in one pattern but not in the other.
The ordering of the variables in the projection must respect the total order imposed in the variables. This guarantees that the attributes are the same and by the same order in the resulting argument expressions of the union operator.

\begin{definition}[Translation of the {\tt AND} pattern] Consider the graph pattern $(P_1  {\tt\ AND\ } P_2)$ and let $var(P_1) \cap var(P_2) = \{ v_1, \ldots, v_n\}$ (which may be empty). The relational algebra expression $\transexpr{(P_1  {\tt\ AND\ } P_2)}^G$\@ is
\[
\begin{array}{l@{}l}
\Pi_{\scriptsize\begin{array}{l}G,\\var(P_1) - var(P_2),\\var(P_2) - var(P_1),\\ v_1 \leftarrow first(v'_1,v''_1),\ldots,\\ v_n \leftarrow first(v'_n,v''_n)\end{array}} & \left[ \sigma_{comp} \left( \rho_{\scriptsize\begin{array}{c}v'_1 \leftarrow v_1\\ \vdots\\ v'_n \leftarrow v_n\\\end{array}} \left(\transexpr{P_1}^G\right)  \Join  \rho_{\scriptsize\begin{array}{c}v''_1 \leftarrow v_1\\ \vdots\\ v''_n \leftarrow v_n\\\end{array}} \left(\transexpr{P_2}^G \right) \right) \right]
\end{array}
\] 
where $comp$ is a conjunction of conditions $v'_i = {\tt unb} \vee v''_i = {\tt unb} \vee v'_i = v''_i$ for each variable $v_i (1 \leq i \leq n)$\@. The function $first$ returns the first argument which is not {\tt unb}\@, or {\tt unb} if both arguments are {\tt unb}\@. 
Note that if the set of common variables is empty then the relational algebra expression simplifies to:
\[
\Pi_{G,var(P_1) \cup var(P_2)} \left[ \transexpr{P_1}^G  \Join \transexpr{P_2}^G \right]
\]
\end{definition}

We need to rename common variables in both arguments, since an unbound variable is compatible with any bound or unbound value  in order to be able to check compatibility using a selection (it is well-known that the semantics of {\tt unb} is different from  semantics of NULLs in relational algebra). The use of the $first$ function in the extended projection is used to obtain in the solution the bound value of the variable, whenever it exists. This technique is the same with that used in~\cite{Elliott:2009,Polleres07}. The use of the extended projection is not essential, since it can be translated into a more complex relational algebra query by using an auxiliary relation containing a tuple for each pair of compatible pairs of variables.

\begin{definition}[Translation of the {\tt MINUS} pattern] Consider the graph pattern $(P_1  {\tt\ MINUS\ } P_2)$ and let $var(P_1) \cap var(P_2) = \{ v_1, \ldots, v_n\}$ (which may be empty). The relational algebra expression  $\transexpr{(P_1  {\tt\ MINUS\ } P_2)}^G$\@ is
\[
\transexpr{P_1}^G  \Join \left[ \delta\left( \transexpr{P_1}^G \right) - \Pi_{G,var(P_1)} \left[ \sigma_{comp  \wedge \neg disj} \left( \transexpr{P_1}^G  \Join  \rho_{\scriptsize\begin{array}{c}v'_1 \leftarrow v_1\\ \vdots\\ v'_n \leftarrow v_n\\\end{array}} \left(\transexpr{P_2}^G \right) \right) \right]\right]
\]
where $comp$ is a conjunction of conditions $v_i = {\tt unb} \vee v'_i = {\tt unb} \vee v_i = v'_i$ for each variable $v_i (1 \leq i \leq n)$\@, and $disj$ is the conjunction of conditions $v_i = {\tt unb} \vee v'_i = {\tt unb}$\@ for each variable $v_i (1 \leq i \leq n)$\@.
Note that if the set of common variables is empty then the above expression reduces to $\transexpr{P_1}^G$ since $disj=true$\@.
\end{definition}

This is the first of the non-monotonic SPARQL patterns, and deserves some extra explanation. We need to check dynamically if the domains of variables are  disjoint since we do not know at translation time what are the unbound variables in the solution mappings, except when trivially the arguments of {\tt MINUS} do not share any variable. The expression on the right hand side of the difference operator returns a tuple corresponding to a solution mapping $\mu_1$ of $P_1$ whenever it is possible to find a solution mapping $\mu_2$ of $P_2$ that it is compatible with $\mu_1$ (condition $comp$) and the mappings do not have disjoint domains (condition $\neg disj$). By deleting these tuples (solutions) from solutions of $P_1$ we negate the condition, and capture the semantics of the {\tt MINUS} operator.
The use of the duplicate elimination $\delta$ ensures that only one tuple is obtained for each solution mapping, in order to guarantee that the cardinality of the result is as what is specified by SPARQL semantics: each tuple in $\transexpr{P_1}^G$ joins with at most one tuple (itself) resulting from the difference operation. 

\begin{definition}[Translation of {\tt FILTER} pattern] Consider the graph pattern $(P {\tt\ FILTER\ } R)$\@, and let ${\tt [NOT]\ EXISTS}(P_1)$\@, \ldots, ${\tt [NOT]\ EXISTS}(P_m)$ the {\tt EXISTS} or {\tt NOT EXISTS} filter expressions occurring in R (which might not occur)\@. The relational algebra expression $\transexpr{(P  {\tt\ FILTER\ } R)}^G$\@ is
\[
\Pi_{G,var(P)}\left[\sigma_{filter} \left( \transexpr{P}^G \Join E_1 \Join \ldots \Join E_m \right)\right]
\]
where $filter$ is a condition obtained from $R$  where each occurrence of ${\tt EXISTS}(P_i)$ (resp. ${\tt NOT\ EXISTS}(P_i)$) is substituted by condition $ex_i <> 0$ (resp. $ex_i = 0$)\@, where $ex_i$ is a new attribute name. Expression $E_i (1 \leq i \leq m)$ is:
\[
\begin{array}{c}
\Pi_{G, var(P),ex_i \leftarrow 0} \left[ \delta(P') - \Pi_{G,var(P)}\left( \sigma_{subst}\left( P' \Join \rho_{\scriptsize\begin{array}{c}v'_1 \leftarrow v_1\\ \vdots\\ v'_n \leftarrow v_n\\\end{array}} (P'_i) \right)\right)\right]\\
\bigcup\\
\Pi_{G, var(P),ex_i \leftarrow 1} \left[ \delta(P') - \left[ \delta(P') - \Pi_{G,var(P)}\left( \sigma_{subst}\left(P' \Join \rho_{\scriptsize\begin{array}{c}v'_1 \leftarrow v_1\\ \vdots\\ v'_n \leftarrow v_n\\\end{array}} (P'_i) \right)\right) \right] \right]
\end{array}
\]
where $P'=\transexpr{P}^G$, $P_i'=\transexpr{P_i}^G$, and $\mathit{subst}$ is the conjunction of conditions $v_i = v'_i \vee v_i = {\tt unb}$ for each variable $v_i$ in $var(P) \cap var(P_i) = \{ v_1, \ldots, v_n\}$\@.
Note that if there are no occurrences of {\tt EXISTS} patterns, then $\transexpr{(P  {\tt\ FILTER\ } R)}^G$ is $\sigma_{R} \left( \transexpr{P}^G \right)$\@.
\end{definition}

The translation of {\tt FILTER} expressions turns out to be very complex due to the {\tt EXISTS} patterns. For each exists expression we need to introduce an auxiliary expression returning a unique tuple for each solution mapping of $P$\@, the top expression when the pattern $P_i$ does not return any solution, and the bottom expression when it does. We need the double negation in order to not affect the cardinality of the results of the filter operation when  pattern $P$ returns more than one solution. Obviously, our translation depends on the capability of expressing arbitrary SPARQL conditions as relational algebra conditions; this is not  immediate but assumed possible due to the translation provided in~\cite{Elliott:2009}.

We can now conclude our translation by taking care of the {\tt OPTIONAL} graph pattern, since it depends on the translation of filter patterns:
\begin{definition}[Translation of {\tt OPTIONAL\ } pattern] Consider the graph pattern $(P_1  {\tt\ OPTIONAL\ } (P_2 {\tt\ FILTER\ } R))$\@.\\ The relational algebra expression  $\transexpr{(P_1  {\tt\ OPTIONAL\ } (P_2 {\tt\ FILTER\ } R))}^G$\@ is
\[
\begin{array}{c}
\transexpr{(P_1  {\tt\ AND\ } P_2)}^G\\
\bigcup\\
\Pi_{G,\scriptsize var(P_1) \cup \{ v \leftarrow {\tt unb} \mid v \in var(P_2) \setminus var(P_1)\}} \\
\\
\left[\transexpr{P_1}^G \Join \left(\begin{array}{c} \delta\left(\transexpr{P_1}^G\right)\\ -\\  \Pi_{G,var(P_1)} \left( \transexpr{(P_1  {\tt\ AND\ } P_2) {\tt\ FILTER\ } R}^G\right) \end{array}\right)\right]
\end{array}
\]
The translation of $(P_1  {\tt\ OPTIONAL\ } P_2)$ is obtained from the translation of the graph pattern $(P_1  {\tt\ OPTIONAL\ } (P_2 {\tt\ FILTER\ } true))$\@.
\end{definition}

The translation of the {\tt OPTIONAL} pattern has two parts, one corresponding to the {\tt JOIN} operator (top expression) and one corresponding to the {\bf Diff}  operator. The translation of the {\bf Diff} operator uses the same technique as the {\tt MINUS} operator but now we remove from solutions of $P_1$ those solution mappings of $P_1$ that are compatible with a mapping of $P_2$ and that satisfy the filter expression.

\begin{theorem}[Correctness of translation]\label{th:map2ra} Given a graph pattern $P$ and a RDF dataset $D(G)$ the process of evaluating the query is performed as follows:
\begin{enumerate}
\item Construct the base relations {\tt Graphs} and {\tt Quads} from $D(G)$\@;
\item\label{expr:query} Evaluate $\transexpr{SPARQL(P,D(G),V)} = \Pi_{V}\left[\sigma_{G'=0}\left( \transexpr{()}^{G'} \Join \transexpr{P}^{G'} \right)\right]$\@ with respect to the base relations {\tt Graphs} and {\tt Quads}, where $G'$ is a new attribute name and $V \subseteq var(P)$\@.
\end{enumerate}
Moreover, the tuples of relational algebra query~(\ref{expr:query}) are in one-to-one correspondence with the solution mappings of $\evalexpr{P}_{D(G)}$ when $V = var(P)$\@, and where an attribute mapped to {\tt unb} represents that the corresponding variable does not belong to the domain of the solution mapping.
\end{theorem}

\begin{proof} The proof is by by structural induction on the graph patterns and can be found in the extended version of this paper available at~\url{http://arxiv.org/abs/1209.0378}.
\end{proof}
The constructed translation will be used to extract how-provenance information for SPARQL queries, addressing the problems identified in~\cite{Theoharis:2011}.

\section{Provenance for SPARQL queries}\label{sec:prov}

The crux of the method has been specified in the previous section, and relies on the properties of the extended provenance $m$-semiring ${\cal K}_{dprovd}$  for  language $\cal{RA}^+_{\cal K}(-,\delta)$\@. We just need a definition before we illustrate the approach.

\begin{definition}[Provenance for SPARQL]\label{def:provsparql} Given a graph pattern $P$ and a RDF dataset $D(G)$ the provenance for $P$ is obtained as follows:
\begin{itemize}
\item Construct the base ${\cal K}_{dprovd}$-relations by annotating each tuple in {\tt Graphs} and {\tt Quads} with a new identifier;
\item Construct an annotated query $SPARQL(P,D(G),V)_{{\cal K}_{dprovd}}$ from relational algebra $\transexpr{SPARQL(P,D(G),V)}$ expression by substituting the duplicate elimination operator by $\delta_{1}$ where $1$ is the identity element of ${{\cal K}_{dprovd}}$\@.
\end{itemize}
The provenance information for $P$ is the annotated relation obtained from evaluating $SPARQL(P,D(G),V)_{{\cal K}_{dprovd}}$ with respect to the annotated translation of the dataset $D(G)$\@.
\end{definition}

By the factorization property of ${\cal K}_{dprovd}$ we know that this is the most general $m$-semiring, and thus the provenance obtained according to Definition~\ref{def:provsparql} is the most informative one. We just need to illustrate the approach with Example~\ref{ex:motopt} in order to completely justify its appropriateness.

\begin{example} First, we represent the RDF dataset by ${\cal K}_{dprovd}$-relations where the annotation tags are shown in the last column. The IRIs have been abbreviated:
\[\small
\begin{array}{l}
\begin{array}{ll||l}
\multicolumn{3}{l}{\tt Graphs}\\
{\tt gid} & {\tt IRI}\\\hline
0 & \iri{} & g_0\\
\end{array}\\
\\
\begin{array}{l@{\quad}lll||l}
\multicolumn{5}{l}{\tt Quads}\\
{\tt gid} & {\tt sub} & {\tt pred} & {\tt obj} & \\\hline
0 & \iri{david} & \iri{account} & \iri{bank} & t_1\\
0 & \iri{felix} & \iri{account} & \iri{games} & t_2\\
0 & \iri{bank} & \iri{accountServiceHomepage} & \iri{bank/yourmoney} & t_3\\
\end{array}\\
\end{array}
\]
Returning to query $Q = (Q_1 {\tt\ OPTIONAL\ } Q_2)$ of Example~\ref{ex:query} with (sub)patterns $Q_1= ( ?w, \iri{account}, ?a )$ and $Q_2 = (?a, \iri{accountServiceHomepage}, ?h )$\@, we obtain the following expressions for $Q_1$ and $Q_2$\@:
\[\small
\begin{array}{lll}
\transexpr{Q_1}^G & = & \Pi_{G,w,a}\left[ \rho_{\scriptsize\begin{array}{l}G \leftarrow {\tt gid}\\ w \leftarrow {\tt sub}\\ a \leftarrow {\tt obj} \end{array}} \left({\scriptsize\sigma_{{\tt pred} = \iri{account} }}({\tt Quads}) \right) \right]\\
\transexpr{Q_2}^G & = & \Pi_{G,a,h}\left[ \rho_{\scriptsize\begin{array}{l}G \leftarrow {\tt gid}\\ a \leftarrow {\tt sub}\\  h \leftarrow {\tt obj}\end{array}} \left({\scriptsize\sigma_{{\tt pred} = \iri{accountServiceHomepage} }}({\tt Quads}) \right) \right]\\
\end{array}
\]
returning the annotated relations:
\[\small
\begin{array}{l}
\transexpr{Q_1}^G=\begin{array}{l@{\quad}ll||l}
G & w & a & \\\hline
0 & \iri{david} & \iri{bank} &  t_1\\
0 & \iri{felix} &  \iri{games} & t_2\\
\end{array}\\
\\
\transexpr{Q_2}^G=\begin{array}{l@{\quad}ll||l}
G & a & h & \\\hline
0 & \iri{bank} & \iri{bank/yourmoney} & t_3\\
\end{array}
\end{array}
\]
The expression $\transexpr{(Q_1 {\tt\ AND\ } Q_2)}^G$ used in the construction of the expression for the {\tt OPTIONAL} pattern is:
\[\small
\Pi_{G,w,a \leftarrow first(a',a''),h} \left[ \sigma_{a'=a'' \vee a'={\tt unb} \vee a''={\tt unb}} \left(  \rho_{a'' \leftarrow a}\left(\transexpr{Q_1}^G\right) \Join \rho_{a' \leftarrow a}\left(\transexpr{Q_2}^G\right)\right)\right]
\]
obtaining the annotated relation:
\[\small
\transexpr{(Q_1 {\tt\ AND\ } Q_2)}^G=\begin{array}{l@{\quad}lll||l}
G & w & a & h\\\hline
0 & \iri{david} & \iri{bank} & \iri{bank/yourmoney} & t_1 \times t_3\\
\end{array}\\
\]
We also need to determine the value of $\delta_1({\transexpr{Q_1}^G})$ which is simply:
\[\small
\delta_{1}(\transexpr{Q_1)}^G)=\begin{array}{l@{\quad}ll||l}
G & w & a & \\\hline
0 & \iri{david} & \iri{bank} &  1\\
0 & \iri{felix} &  \iri{games} & 1\\
\end{array}
\]
We can now construct the expression corresponding to the {\bf Diff} operator of SPARQL algebra, namely:
\[\small
\begin{array}{c}
\Pi_{G,w,a,h\leftarrow {\tt unb}}
\left[\transexpr{Q_1}^G \Join \left(\begin{array}{c} \delta_{1}\left(\transexpr{Q_1}^G\right)\\ -\\  \Pi_{G,w,a} \left( \transexpr{(Q_1  {\tt\ AND\ } Q_2)}^G\right) \end{array}\right)\right]
\end{array}
\]
returning the annotated tuples:
\[\small
\begin{array}{l@{\quad}lll||l}
G & w & a & h & \\\hline
0 & \iri{david} & \iri{bank} &  {\tt unb} & t_1 \times(1 -( t_1 \times t_3))\\
0 & \iri{felix} &  \iri{games} & {\tt unb} & t_2 \times(1 - 0) = t_2\\
\end{array}
\]
This is the important step, since $K$-relations assign an annotation to every possible tuple in the domain. If it is not in the support of the relation, then it is tagged with $0$\@. Therefore, the solutions for $(\transexpr{(Q_1  {\tt\ AND\ } Q_2)}^G$ are:
\[\small
\begin{array}{l@{\quad}lll||l}
G & w & a & h & \\\hline
0 & \iri{david} & \iri{bank} &  \iri{bank/yourmoney} & t_1 \times t_3\\
0 & \iri{david} & \iri{bank} &  {\tt unb} & t_1 \times(1 - ( t_1 \times t_3))\\
0 & \iri{felix} &  \iri{games} & {\tt unb} & t_2\\
\end{array}
\]
and for our query, finally we get
\[\small
\begin{array}{l@{\quad}ll||l}
w & a & h & \\\hline
\iri{david} & \iri{bank} &  \iri{bank/yourmoney} & g_0 \times t_1 \times t_3\\
\iri{david} & \iri{bank} &  {\tt unb} & g_0 \times t_1 \times(1 - ( t_1 \times t_3))\\
\iri{felix} &  \iri{games} & {\tt unb} &  g_0 \times t_2\\
\end{array}
\]
The interpretation of the results is the expected and intuitive one. Suppose that (i) we use the boolean $m$-semiring, with just the two values {\tt t} and {\tt f}, meaning that we trust or not trust a triple, (ii) product corresponds to conjunction, (iii) sum corresponds to disjunction, and (iv) difference is defined as $x - y=x \wedge \neg y$\@. So, if we trust $g_0$ and $t_1$\@, $t_2$ and $t_3$  we are able to conclude that we trust the first and third solutions (substitute $1$ and the identifiers of trusted triples by {\tt t} in the annotations, and then evaluate the resulting boolean expression). If we do not trust $t_3$ but trust the other triples then we trust the second and third solutions. Also mark how the graph provenance is also annotated in our solutions. Accordingly, if we don't trust the default graph then we will not trust any of the solutions. Therefore, our method was capable of keeping in the same annotated ${\cal K}_{dprovd}$-relation the several possible alternative solutions, one in each distinct tuple. This was claimed to not be possible in~\cite{Theoharis:2011}.
\end{example} 
  
\section{Discussion and Conclusions}

The literature describes several approaches to extract data provenance/annotated information from RDF(S) data~\cite{Flouris:2009,Dividino:2009,Buneman:2010,Theoharis:2011,Zimmermann12}. A first major distinction is that we extract how-provenance instead of only why-provenance\footnote{We use the terminology ``how-'' and ``why-provenance'' in the sense of~\cite{GKT07:pods}.} of~\cite{Flouris:2009,Dividino:2009,Buneman:2010,Zimmermann12}. Both~\cite{Flouris:2009,Buneman:2010} address the problem of extracting data provenance for RDF(S) entailed triples, but do not support SPARQL. The theory developed in~\cite{Dividino:2009} implements the difference operator using a negation, but it does not handle duplicate solutions according to the semantics of SPARQL because of idempotence of sum; additionally, the proposed difference operator to handle why-provenance discards the information in the right hand argument. The most complete work is~\cite{Zimmermann12} which develops a framework for annotated Semantic Web data, supporting RDFS entailment and providing a query language extending many of the SPARQL features in order to deal with annotated data, exposing annotations at query level via annotation variables, and including aggregates and subqueries (but not property path patterns). However, the sum operator is idempotent in order to support RDFS entailment, and by design the {\tt UNION} operator is not interpreted in the annotation domain. Moreover, the {\tt OPTIONAL} graph pattern discards in some situations the information in the second argument, and thus cannot extract full provenance information.  

The capability of extracting full data how-provenance for SPARQL semantics as prescribed in~\cite{Theoharis:2011} has been shown possible with our work, refuting their claim that existing algebras could not be used for SPARQL. Our approach, like~\cite{Theoharis:2011}, rests on a translation of SPARQL into annotated relational algebra contrasting with the abstract approach of~\cite{Flouris:2009,Dividino:2009,Buneman:2010,Zimmermann12}. The authors in~\cite{Theoharis:2011} argue that this translation process does not affect the output provenance information for the case of (positive) SPARQL. In this way, the major constructs of SPARQL 1.1 are taken care respecting their bag semantics. However, contrary to the works of~\cite{Flouris:2009,Zimmermann12} we do not address the RDF schema entailment rules, and therefore our work is only applicable to simple entailment. 

We plan to address the complete semantics of SPARQL. In particular, aggregates can be handled by summing ($\oplus$) tuples for each group, while property path patterns can generate annotation corresponding to products ($\otimes$) of the involved triples in each solution.
This extension is enough to be able to capture data provenance for RDFS entailment. We also want to explore additional applications in order to assess fully the potential of the  proposed method. 

\paragraph{Acknowledgments} C. V. Damásio was partially supported by FCT Project ERRO PTDC/EIACCO/121823/2010.

\bibliographystyle{abbrv}
\bibliography{provlp}

\newpage

\appendix

\section{Proof of the main result}

This appendix provides a detailed proof of Theorem~\ref{th:map2ra} by structural induction on the graph patterns.

The induction proof will construct an expression containing as attributes the graph attribute and an attribute for each in-scope variable of the graph pattern. Assume that the active graph is given and has id $j$ ($0$ for the default graph, or $1 \leq j \leq n$ for the case of a named graph). We show that the cardinality of solutions of each SPARQL algebra operator is respected for the active graph, which is the most difficult part. Notice that the graphs by definition do not have duplicate triples, and therefore there will not be duplicates in the original base relations {\tt Graphs} and {\tt Quads}.

\paragraph{Empty graph pattern} Recall that he empty graph pattern is translated into $\transexpr{()}^G = \Pi_G\left[  \rho_{G \leftarrow {\tt gid}} ({\tt Graphs}) \right]$\@. This means that the result of the relational algebra query returns as many as tuples as graphs in the RDF dataset, and at least a tuple with identifier 0 for the initial default graph. Moreover, note that for each particular graph  id (default or named) this relation has exactly 1 tuple with that id.

\paragraph{Triple pattern} Since there are no duplicate triples in the graphs, then the number of possible solutions of a triple pattern are exactly the number of triples that match the pattern, one for each instance. We analyse the correctness of the translation of the triple pattern according to the number of variables  that may occur on it: no variables, one variable, two variables and three variables.

\begin{description}
\item[0 vars:] In this case the triple pattern  $t=(t_1,t_2,t_3)$ contains only RDF terms which might be identical, or not. In this case we obtain the relational algebra expression $\Pi_G\left[\rho_{G\leftarrow {\tt gid}}\left(\sigma_{{\tt sub} = t_1 \wedge {\tt pred} = t_2 \wedge {\tt obj} = t_3}({\tt Quads})\right)\right]$\@. The resulting relation has just the attribute (column) $G$\@. If the triple occurs in graph $i$ then the resulting relation will contain a tuple having value $i$ in attribute $G$\@; otherwise no tuple will occur for graph $i$\@. 
\item[1 var:] We have several cases to consider here: there is just one occurrence of a variable, two or three.

Consider that the variable occurs in the subject of the triple; for the remaining cases in predicate or object the reasoning is similar. So, let $t=(?v,t_2,t_3)$, obtaining the relational algebra expression \[\Pi_{G,v}\left[\rho_{G\leftarrow {\tt gid}, v \leftarrow {\tt sub}}\left(\sigma_{{\tt pred} = t_2 \wedge {\tt obj} = t_3}({\tt Quads})\right)\right]\] The resulting relation has two columns, one for the graph $G$ and one for collecting the bindings for $v$ (i.e. the solution). If the triple occurs in graph $i$ then the resulting relation will contain a tuple having value $i$ in attribute $G$\@; otherwise no tuple will occur for graph $i$\@. 

Consider now that the variable occurs in the subject and predicate of the triple; for the remaining cases the reasoning is similar. So, let $t=(?v,?v,t_3)$, obtaining the relational algebra expression \[\Pi_{G,v}\left[\rho_{G\leftarrow {\tt gid}, v \leftarrow {\tt sub}}\left(\sigma_{{\tt sub} ={\tt pred} \wedge {\tt obj} = t_3}({\tt Quads})\right)\right]\] The resulting relation has two columns, one for the graph $G$ and one for collecting the bindings for $v$ (i.e. the solution). The selection condition guarantees that the obtained instances match the triple pattern, obtaining for each graph $i$ one tuple for each possible value of $v$\@.

If there are three occurrences of the variable $?v$ then $t=(?v,?v,?v)$, then one obtains the expression  \[\Pi_{G,v}\left[\rho_{G\leftarrow {\tt gid}, v \leftarrow {\tt sub}}\left(\sigma_{{\tt sub} ={\tt pred} \wedge {\tt sub} ={\tt obj} \wedge {\tt pred} ={\tt obj}}({\tt Quads})\right)\right]\] Note that one of the equalities in the selection condition is redundant. As before, for each graph, we get a tuple for each possible value of $v$\@.

\item[2 vars:] There are two variant of patterns here $t=(?v1,?v2,t_3)$ and $t=(?v1,?v2,?v2)$ and permutations. For the first pattern we get
 \[\Pi_{G,v1,v2}\left[\rho_{G\leftarrow {\tt gid}, v1 \leftarrow {\tt sub},v2 \leftarrow {\tt pred}}\left(\sigma_{{\tt obj} = t_3}({\tt Quads})\right)\right]\] while  for the second we get
\[\Pi_{G,v1,v2}\left[\rho_{G\leftarrow {\tt gid}, v1 \leftarrow {\tt sub},v2 \leftarrow {\tt pred}}\left(\sigma_{{\tt pred} = {\tt obj}}({\tt Quads})\right)\right]\]
It is easy to see that each solution (for each graph) corresponds exactly to one tuple in the evaluation of the translated  relation.
\item[3 vars:] This case is immediate and being the triple pattern $t=(?v1,?v2,?v3)$\@. The translation is $\Pi_{G,v_1,v_2,v_3}\left[\rho_{G\leftarrow {\tt gid},v1\leftarrow {\tt sub},v2 \leftarrow {\tt pred},v3\leftarrow {\tt obj}}({\tt Quads})\right]$\@, obtaining a relation which is isomorphic to ${\tt Quads}$, as expected. For each graph $i$, we obtain a tuple with value $i$ in attribute $G$ and remaining attributes corresponding exactly to one triple in graph $i$\@.
\end{description}

\paragraph{{\tt UNION} pattern} The translated relational algebra expression $\transexpr{(P_1  {\tt\ UNION\ } P_2)}^G$  is:
\[
\begin{array}{c}
\Pi_{G,var(P_1) \cup \{ v \leftarrow {\tt unb} \mid v \in var(P_2) \setminus var(P_1)\}} \left( \transexpr{P_1}^G \right)\\
\bigcup\\
\Pi_{G,var(P_2) \cup \{ v \leftarrow {\tt unb} \mid v \in var(P_1) \setminus var(P_2)\}} \left( \transexpr{P_2}^G \right)\\
\end{array}
\]  
The relational algebra expressions makes the union of two projections. Each projection will not remove any in-scope variable, and it is used to extend the columns with unbound values in order to obtain a relation with columns $G \cup var(P_1) \cup var(P_2)$\@, by making unbound the variables that do not occur in the sub-pattern $\transexpr{P_1}^G$ or  $\transexpr{P_2}^G$\@, respectively. Therefore, each subexpression  the projection operator will not remove any in-scope attributes of each sub-pattern $\transexpr{P_1}^G$ or  $\transexpr{P_2}^G$, and thus it returns exactly as many as tuples as the number of solutions of each sub-pattern by induction hypothesis. The cardinality of the resulting expression  is the sum of solutions of each sub-expression, according to the bag semantics of relational algebra.

\paragraph{ {\tt AND pattern}} The translated relational algebra expression $\transexpr{(P_1  {\tt\ AND\ } P_2)}^G$  is:
\[
\begin{array}{l@{}l}
\Pi_{\scriptsize\begin{array}{l}G,\\var(P_1) - var(P_2),\\var(P_2) - var(P_1),\\ v_1 \leftarrow first(v'_1,v''_1),\ldots,\\ v_n \leftarrow first(v'_n,v''_n)\end{array}} & \left[ \sigma_{comp} \left( \rho_{\scriptsize\begin{array}{c}v'_1 \leftarrow v_1\\ \vdots\\ v'_n \leftarrow v_n\\\end{array}} \left(\transexpr{P_1}^G\right)  \Join  \rho_{\scriptsize\begin{array}{c}v''_1 \leftarrow v_1\\ \vdots\\ v''_n \leftarrow v_n\\\end{array}} \left(\transexpr{P_2}^G \right) \right) \right]
\end{array}
\] 
where $comp$ is a conjunction of conditions $v'_i = {\tt unb} \vee v''_i = {\tt unb} \vee v'_i = v''_i$ for each common variable $v_i (1 \leq i \leq n)$\@. The function $first$ returns the first argument which is not {\tt unb}\@, or {\tt unb} if both arguments are {\tt unb}\@. 

The joined subexpressions inside the selection have only the common attribute $G$ due to the renaming of common variables. So, every tuple of the subexpression for $\transexpr{P_1}^G$ will join 
with every tuple of $\transexpr{P_2}^G$, for each graph. If there is a solution with cardinality $c_1$ of $P_1$ and solution with cardinality $c_2$ of $P_2$, one will obtain $c_1 \times c_2$ tuples in the result of the joined expression. From these possibilities, the selection expression keeps he combinations of solutions which are compatible: for each pair $v'_i$ and $v''_i$ the condition guarantees that at least 1 variable is unbound or have the same value. So, we only keep the tuples for each possible merge (with cardinality $c_1 \times c_2$ since selection keeps duplicate tuples). The projection is necessary to obtain the relation on the original variables, besides $G$, obtaining the value from the first bound variable, if any. It is also necessary to recall that the bag semantics for the projection operator will obtain as many tuples as the contributing tuples, summing over the obtained solutions  as required by the cardinality condition of the {\tt AND} operator (see the definition of the $\cup$ for $K$-relations) .

\paragraph{ {\tt FILTER pattern}} The relational algebra expression $\transexpr{(P  {\tt\ FILTER\ } R)}^G$\@ is
\[
\Pi_{G,var(P)}\left[\sigma_{filter} \left( \transexpr{P}^G \Join E_1 \Join \ldots \Join E_m \right)\right]
\]
where $filter$ is a condition obtained from $R$  where each occurrence of ${\tt EXISTS}(P_i)$ (resp. ${\tt NOT\ EXISTS}(P_i)$) in $R$ is substituted by condition $ex_i <> 0$ (resp. $ex_i = 0$)\@, where $ex_i$ is a new attribute name. Expression $E_i (1 \leq i \leq m)$ is:
\[
\begin{array}{c}
\Pi_{G, var(P),ex_i \leftarrow 0} \left[ \delta(P') - \Pi_{G,var(P)}\left( \sigma_{subst}\left( P' \Join \rho_{\scriptsize\begin{array}{c}v'_1 \leftarrow v_1\\ \vdots\\ v'_n \leftarrow v_n\\\end{array}} (P'_i) \right)\right)\right]\\
\bigcup\\
\Pi_{G, var(P),ex_i \leftarrow 1} \left[ \delta(P') - \left[ \delta(P') - \Pi_{G,var(P)}\left( \sigma_{subst}\left(P' \Join \rho_{\scriptsize\begin{array}{c}v'_1 \leftarrow v_1\\ \vdots\\ v'_n \leftarrow v_n\\\end{array}} (P'_i) \right)\right) \right] \right]
\end{array}
\]
where $P'=\transexpr{P}^G$, $P_i'=\transexpr{P_i}^G$, and $\mathit{subst}$ is the conjunction of conditions $v_i = v'_i \vee v_i = {\tt unb}$ for each variable $v_i$ in $var(P) \cap var(P_i) = \{ v_1, \ldots, v_n\}$\@.

The translation is  complex due to the {\tt EXISTS} expressions. Note that each $E_i$ expression returns exactly one tuple for each solution of pattern {\tt P}. Thus, the duplicate removal operations are there just to guarantee this and do no affect the cardinality of the solutions, obtaining one solution for each solution of {\tt P} that obeys to the filter condition. The rest of the translation is more or less immediate, where the condition $subst$ does not correspond exactly to the compatibility condition used before, since according to the SPARQL semantics the variables of pattern ${\tt P_1}$ are substituted in the {\tt EXISTS} pattern (we discard the cases where a variable is bound by ${\tt P_1}$ and not bound in ${\tt P_i}$)\@.

\paragraph{ {\tt MINUS pattern}} The relational algebra expression  $\transexpr{(P_1  {\tt\ MINUS\ } P_2)}^G$\@ is
\[
\transexpr{P_1}^G  \Join \left[ \delta\left( \transexpr{P_1}^G \right) - \Pi_{G,var(P_1)} \left[ \sigma_{comp  \wedge \neg disj} \left( \transexpr{P_1}^G  \Join  \rho_{\scriptsize\begin{array}{c}v'_1 \leftarrow v_1\\ \vdots\\ v'_n \leftarrow v_n\\\end{array}} \left(\transexpr{P_2}^G \right) \right) \right]\right]
\]
where $comp$ is a conjunction of conditions $v_i = {\tt unb} \vee v'_i = {\tt unb} \vee v_i = v'_i$ for each variable common $v_i (1 \leq i \leq n)$\@, and $disj$ is the conjunction of conditions $v_i = {\tt unb} \vee v'_i = {\tt unb}$\@ for each variable $v_i (1 \leq i \leq n)$\@.

The difference expression will return either a tuple of $\transexpr{P_1}^G$ or not, without duplicates. The difference expression evaluation will return solutions of $P_1$ that do not belong to the expression in the right-hand side of the difference. Since the right-hand side of expression obtains the tuples corresponding to the solutions of ${\tt P_1}$ for which there is at least one solution in ${\tt P_2}$ that is compatible and not disjoint (no bound variable in common), then we obtain as result the tuples corresponding to solutions of $P_1$ such that for all solution $P_2$ the solutions are incompatible or are disjoint (the semantics of {\tt MINUS}). Now, the difference expression will have no duplicates, and thus we keep only the solutions of $P_1$ that join (i.e. that obey to the condition), without increasing the cardinality of the result.

\paragraph{{\tt OPTIONAL pattern}} The relational algebra  expression $\transexpr{(P_1  {\tt\ OPTIONAL\ } (P_2 {\tt\ FILTER\ } R))}^G$\@ is
\[
\begin{array}{c}
\transexpr{(P_1  {\tt\ AND\ } P_2)}^G\\
\bigcup\\
\Pi_{G,\scriptsize var(P_1) \cup \{ v \leftarrow {\tt unb} \mid v \in var(P_2) \setminus var(P_1)\}} \\
\\
\left[\transexpr{P_1}^G \Join \left(\begin{array}{c} \delta\left(\transexpr{P_1}^G\right)\\ -\\  \Pi_{G,var(P_1)} \left( \transexpr{(P_1  {\tt\ AND\ } P_2) {\tt\ FILTER\ } R}^G\right) \end{array}\right)\right]
\end{array}
\]
According to the semantics of SPARQL 1.1, the {\tt OPTIONAL} pattern evaluation is performed by two operators: a join and a left join. This is particularly clear in the translation, where the join is the first expression and the left join the lower (big) expression below the union operator. The rationale of the translation of the left join operator is identical to the translation of the {\tt MINUS} pattern, except now that we need to make unbound the variables in ${\tt P_2}$ but not in ${\tt P_1}$. Again, the number of obtained tuples is according to the semantics of SPARQL 1.1: it is the sum of the tuples of the join with the sum of tuples of the left join (guaranteed by the bag semantics of $\cup$).

\paragraph{{\tt GRAPH pattern} } The translation of $({\tt GRAPH\ } term\ P_1)$ has two cases:
\begin{itemize}
\item If $term$ is an IRI then $\transexpr{({\tt GRAPH\ } term\ P_1)}^G$ is
\[
\transexpr{()}^G \Join \Pi_{var(P_1)} \left[ \Pi_{G'} \left( \rho_{\scriptsize G' \leftarrow {\tt gid}} \left( \sigma_{term = {\tt IRI}}( {\tt Graphs}) \right) \right)\Join \transexpr{P_1}^{G'} \right]
\]
The selection expression obtains a single tuple containing the identifier corresponding to the $term$\@, or obtains the empty relation if there is no named graph with that IRI. This tuple joins with $\transexpr{P_1}^{G'}$ to limit the results to the intended graph. Moreover it is important that the use of a renamed attribute for the graph is necessary in order to avoid clashes of attributes corresponding to the active graph. Moreover, the graph pattern will return the same results independently of the active graph and this is captured by the join with the empty graph pattern.

\item If $term$ is a variable $v$ then $\transexpr{({\tt GRAPH\ } term\ P_1)}^G$ is
\[
\transexpr{()}^G \Join \Pi_{\{v\} \cup var(P_1)} \left[ \rho_{\scriptsize G' \leftarrow {\tt gid}, v \leftarrow {\tt IRI}} \left( \sigma_{{\tt gid} > 0}( {\tt Graphs}) \right)\Join \transexpr{P_1}^{G'} \right]
\]
This case is a little more complex, because now we consider all the named graphs (${\tt gid} > 0$, and bind the variable $v$ with the corresponding IRIs. The rationale of the construction is the same of the previous case.
\end{itemize}

To conclude the proof, we just need to analyse the results of the expression corresponding to the full query $\transexpr{SPARQL(P,D(G),V)} = \Pi_{V}\left[\sigma_{G'=0}\left( \transexpr{()}^{G'} \Join \transexpr{P}^{G'} \right)\right]$\@ with respect to the base relations {\tt Graphs} and {\tt Quads}, where $G'$ is a new attribute name and $V \subseteq var(P)$\@. The selection starts the evaluation at the default graph ($G'=0$) and projects in the selected variables. The correctness of the translation is now immediate due to the previous induction.

\end{document}